\documentclass[12pt]{article}
\usepackage{amsmath}
\usepackage{graphicx,psfrag,epsf}
\usepackage{enumerate}
\usepackage{natbib}
\usepackage{amsmath}
\usepackage{amssymb}
\usepackage{algorithm}
\usepackage{bbm}
\usepackage{fancyhdr}
\usepackage{amsfonts}
\usepackage{subfig,xcolor}
\usepackage{hyperref}
\usepackage{amsthm}
\usepackage{wrapfig}
\usepackage{mathtools}  
\usepackage{enumitem}
\mathtoolsset{showonlyrefs}
\usepackage[noend]{algpseudocode}
\newtheorem{theorem}{Theorem}[section]
\newtheorem{lemma}[theorem]{Lemma}
\newtheorem{definition}[theorem]{Definition}
\usepackage{hyperref}
\usepackage{authblk}

\date{}

\begin{document}

\def\spacingset#1{\renewcommand{\baselinestretch}%
{#1}\small\normalsize} \spacingset{1}


\title{\bf Modelling Correlated Bernoulli Data Part II: Inference}
\author[1]{Louise Kimpton}
\author[1]{Peter Challenor}
\author[2]{Henry Wynn}
\affil[1]{University of Exeter}
\affil[2]{London School of Economics}
  \maketitle

\bigskip
\begin{abstract}
Binary data are highly common in many applications, however it is usually modelled with the assumption that the data are independently and identically distributed.  This is typically not the case in many real-world examples and such the probability of a success can be dependent on the outcome successes of past events.  The de Bruijn process (DBP) was introduced in \cite{Kimpton2022}.  This is a correlated Bernoulli process which can be used to model binary data with known correlation.  The correlation structures are included through the use of de Bruijn graphs, giving an extension to Markov chains. Given the DBP and an observed sequence of binary data, we present a method of inference using Bayes' factors. Results are applied to the Oxford and Cambridge annual boat race.
\end{abstract}

\noindent%
{\it Keywords: De Bruijn Graph, binary, Bernoulli, correlation, Markov chains}

\section{Introduction}

Binary data are highly common in many applications including climate modelling, ecology and genomics. For example, ice sheet data can consist of spatial points on a grid, where the outcome at each location gives either the presence or absence of ice (denoted 1 or 0). With rising surface air temperatures in areas around the globe, ice sheet modelling \citep{Chang2015} is important for predicting future sea level rise due to the melting of ice.  Alternatively in mathematical ecology, there is a significant interest in the measurement of aggregation of a population's spatial pattern \citep{Pielou1984, Pielou1969}. Again, the data may consist of the presence or absence of the individuals in the population (1 or 0), and we may want to model clusters or patch sizes of trees in a forest, or estimate the maximum size for a collection of bacteria in an experiment. 

Typically, binary random variables are treated to be independently and identically distributed with a single probability of success.  However, the assumption that all variables are truly independent is often not sufficient in many applications.  Instead,  there is often a level of correlation between variables in a neighbourhood, such that the probability of a success is dependent on the outcome successes of past events (either spatial or temporal).  If modelling ice sheets or trees and bacteria on a one-dimensional line as above, there is a much higher chance of observing future ice, trees or bacteria when they have already been located in the same area.

Current methods of modelling correlated Bernoulli data include logistic regression and classification methods \citep{Hilbe2009, Kleinbaum1994, Chang2015} as well as generalised linear models \citep{Diggle1998}.  The main issues with these methods are that they only consider marginal distributions when drawing samples, forcing any model outputs to be conditionally independent. There have also been attempts in producing a correlated Bernoulli distribution \citep{Teugels1990, Society2017a, Society2017} as well as considering how to use a multivariate Bernoulli distribution to estimate the structure of graphs with binary nodes \cite{Dai2013}.  However, each of these approaches have to contend with high numbers of parameters. Similarly, there are other approaches involving graphs including \cite{Banerjee2008} and the Ising model, which uses undirected graphs to model groups of binary random variables with more than two interactions \citep{Cipra1987, Ravikumar2010}.

\cite{Kimpton2022} introduced a de Bruijn process. This is a novel framework for modelling and simulating correlated binary trials.  To capture the distance correlation between variables, de Bruijn graphs are used  \citep{Woude1946, Good1946, Golomb1967, Fredricksen1992}. Given the set of symbols, $V=\{0, 1\}$, the vertices or nodes of the graph consist of all the possible sequences of $V$.  These symbols are defined as the 'letters' of the de Bruijn graph, and the sequences of these letters are defined as de Bruijn 'words' of length $m$.  Two de Bruijn graphs of length $m=2$ and $m=3$ are given in Figure \ref{DBG23}.  The nodes consist of all the possible length two or three sequences of $0$’s and $1$’s where there are two edges coming in and out of each node to give a total of $2m$ nodes and $2m + 1$ edges.

\begin{figure}[ht]
\centering
\includegraphics[scale=1]{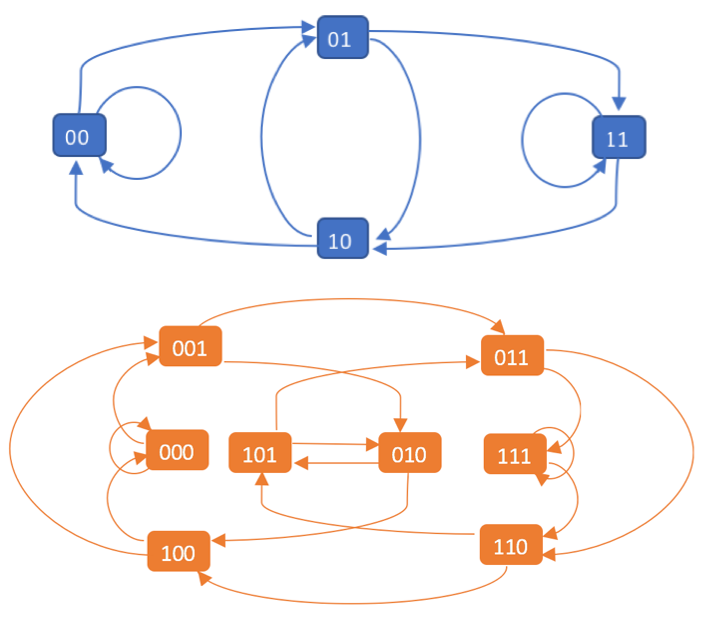}
\caption{Examples of length 2 and 3 de Bruijn graphs with two letters: 0 and 1.}
\label{DBG23}
\end{figure}

As described in \cite{Kimpton2022}, to each of the graph edges in Figure \ref{DBG23}, a probability can be assigned to give the probability of transitioning from each de Bruijn word to the next.  The notation $p_{i}^{j}$ is used to denote the probability of transitioning from the word $i$ to the word $j$. Let $X = \{X_{1}, X_{2}, \ldots, X_{n}\}$ be a length $n$ sequence of binary random variables such that $X_{i}$ is contained within the set of $s=2$ letters, $V = \{0,1\}$. The sequence $X$ can be written in terms of words of length $m$, $X = \{W_{1}, W_{2}, \ldots, W_{n-m+1}\}$, where each $W_{i}$ is an overlapping sequence of $m$ letters.  We thus formally define a de Bruijn process in Definition \ref{DBPdef}. 
\begin{definition}[de Bruijn Process (DBP)] \label{DBPdef}
\rm{Let random variable words $W_{1}, W_{2}, \ldots$ consist of length $m$ sequences of `letters' from the set $V=\{0,1\}$.  For any positive integer, $t$, and possible states (words), $i_{1}, i_{2}, \ldots, i_{t}$, $X$ is described by a stochastic process in the form of a Markov chain with,}
\begin{equation}
\begin{split}
P \left( W_{t} = i_{t} | W_{t-1} = i_{t-1}, W_{t-2} = i_{t-2}, \ldots, W_{1} = i_{1} \right) &= P \left( W_{t} = i_{t} | W_{t-1} = i_{t-1} \right) \\
&= p_{i_{t-1}}^{i_{t}},
\end{split}
\end{equation}
\rm{where $p_{i_{t-1}}^{i_{t}}$ is the probability of transitioning from the word $i_{t-1}$ to word $i_{t}$.}
\end{definition}

Since $X$ has a Markov property on the word and not the letter, far more structure can be incorporated into the sequence.  The authors further describe relationships regarding the stationary distribution of a binary de Bruijn graph, as well as a way of measuring the clustering of a sequence by analysing the run lengths of the letters.  A run length length $R$ was defined as the number of consecutive $1$'s (or $0$'s) in a row bounded by a 0 (or a 1) at both ends. The run length distribution then gives the probability of a run of length $n$ for any $n \in \mathbb{N}^{+}$ in terms of the word length $m$ and the transition probabilities.

There are two main aims of this paper to further develop the de Bruijn process. The first one of these is to fully define a correlated Bernoulli distribution using de Bruijn graphs. Given a sequence of binary random variables, $X = \{ X_{1}, X_{2}, \ldots, X_{n} \}$, we will give an expression for the joint probability of X:
\begin{equation}
\pi^{n}(x_{1} x_{2} \ldots x_{n}) = P(X_{1} = x_{1}, X_{2} = x_{2}, \ldots, X_{n} = x_{n}),
\end{equation}
in terms of the de Bruijn word length, marginal probabilities of the words and associated transition probabilities.  

The second aim of this paper is to develop a method of inference. Each process is fully defined through the length of the word $m$ at each node, and the corresponding transition probabilities. Therefore to predict the de Bruijn process an observed sequence was most likely generated from, we must estimate both of these parameters. Since $m$ can only take integer values, we have chosen to proceed by comparing models using Bayes' factors. The likelihood of the sequence in terms of the transition probabilities can be stated, and so we can further estimate the transition probabilities through maximum likelihood (or Bayesian equivalent).

The remainder of this paper is organised as follows.  Section \ref{CBD} defines the correlated Bernoulli distribution using de Bruijn graphs.  Section \ref{Inf} then gives details on our method of inference to estimate both de Bruijn word lengths and transition probabilities. This also includes examples. Section \ref{App} shows our worked application on the Oxford and Cambridge university boat race. Finally conclusions and future work are discussed in Section \ref{FW}.

\section{Correlated Bernoulli Distribution} \label{CBD}

Let $X = \{X_{1}, X_{2}, \ldots, X_{n}\}$ be a length $n$ sequence of binary random variables which can be modelled using a de Bruijn process with word length $m$ and transition probabilities $p_{i}^{j}$.  A correlated Bernoulli distribution $\pi^{n}$ of sequence $X$ is given in Theorem \ref{CBD1}, which is equivalent to the joint distribution of all binary variables.  As part of the theorem, $\pi^{m}(i)$ denotes the marginal probability of obtaining the word $m$ as part of the stationary distribution of the process. Due to the de Bruijn framework, this expression is given in terms of the word length, marginal probabilities of the de Bruijn words and the transition probabilities. To simplify the notation and to apply to a general word length, the words have been written in terms of the decimal representation of the binary values. This is expressed as $\sum_{i=1}^{m} k_{i} \hspace{0.1cm} 2^{i-1}$, where $k_{i} \in \{0,1\}$ is each letter in the word.

\begin{theorem}[Correlated Bernoulli Distribution $n \ge m$] \label{CBD1}
\rm{For correlated binary random sequence,} $X = \{X_{1}, X_{2}, \ldots, X_{n}\}$\rm{, where} $x= \{x_{1}, x_{2}, \ldots, x_{n}\}$ \rm{is a realisation from} $X$\rm{, the joint probability density is given as follows:}
\begin{equation}
\begin{split}
\pi^{n}(x_{1} & x_{2} \ldots x_{n}) \\
&= P(X_{1} = x_{1}, X_{2} = x_{2}, \ldots, X_{n} = x_{n}) \\
&= \pi^{n}(0...0)^{(1-x_{1})...(1-x_{n})} \times \pi^{n}(0...01)^{(1-x_{1})...(1-x_{n-1})(x_{n})} \times \ldots \\
&\hspace{1cm} \times \pi^{n}(1...10)^{(x_{1})...(x_{n-1})(1-x_{n})} \times \pi^{n}(1...1)^{(x_{1})...(x_{n})} \\
&= \prod_{i=0}^{2^{n}-1} \bigg( \pi^{n}(i) \bigg)^{\prod_{j=1}^{n} \bigg[ \big( x_{j} \big) ^{\big[\frac{1}{2^{n-j}} (i - (i \hspace{0.1cm} \text{mod} 2^{n-j})) \big] \text{mod} 2 } \big( 1 - x_{j} \big) ^{\big[\frac{1}{2^{n-j}} ((2^{n}-i-1) - ((2^{n}-i-1) \hspace{0.1cm} \text{mod} 2^{n-j})) \big] \text{mod} 2 } \bigg]}
\end{split}
\end{equation}
\rm{where}
\begin{equation}
\begin{split}
\pi^{n}(i) &= \sum_{j=0}^{2^{m}-1} \prod_{k=0}^{m-1} \pi^{m}(j) \hspace{0.2cm} p^{2^{k+1} \big( j \hspace{0.1cm}  \text{mod} 2^{m-k-1}\big) + \big[ \frac{1}{2^{n-k-1}} (i - (i \hspace{0.1cm} \text{mod} 2^{n-k-1})) \big] \text{mod} 2^{m} } _{2^{k} \big( j \hspace{0.1cm}  \text{mod} 2^{m-k} \big) + \big[ \frac{1}{2^{n-k}} (i - (i \hspace{0.1cm} \text{mod} 2^{n-k})) \big] \text{mod} 2^{m} } \\
& \hspace{1cm} \times \prod_{s=m}^{n-1} p_{\big[ \frac{1}{2^{n-s}} (i - (i \hspace{0.1cm} \text{mod} 2^{n-s})) \big] \text{mod} 2^{m}} ^{\big[ \frac{1}{2^{n-s-1}} (i - (i \hspace{0.1cm} \text{mod} 2^{n-s-1})) \big] \text{mod} 2^{m} }
\end{split}
\end{equation}
\end{theorem}

\begin{proof}
See Appendix
\end{proof}

Consider a simple example where $n=3$ and $m=2$.  By regarding all possible sequences of letters of length $n=3$, the distribution can be expressed as a product of the marginal probabilities of these sequences as follows:
\begin{equation}
\begin{split}
\pi^{3}(x_{1},x_{2},x_{3}) &= \pi^{3}(000)^{(1-x_{1})(1-x_{2})(1-x_{3})} \times \pi^{3}(001)^{(1-x_{1})(1-x_{2})(x_{3})} \times \pi^{3}(010)^{(1-x_{1})(x_{2})(1-x_{3})} \\
& \hspace{0.5cm} \times \pi^{3}(011)^{(1-x_{1})(x_{2})(x_{3})} \times \pi^{3}(100)^{(x_{1})(1-x_{2})(1-x_{3})} \times \pi^{3}(101)^{(x_{1})(1-x_{2})(x_{3})} \\
& \hspace{0.5cm} \times \pi^{3}(110)^{(x_{1})(x_{2})(1-x_{3})} \times \pi^{3}(111)^{(x_{1})(x_{2})(x_{3})} ,
\end{split}
\end{equation}
for any $x_{i} \in \{0, 1\}, $ $ i=1,2,3$. To express each $\pi^{3}(x_{1} x_{2} x_{3})$ in terms of the length $m$ de Bruijn process,  any initial boundary conditions must be included.  All transitions start with an existing word, this includes taking account of all possible starting words. The law of total probability is used as follows:
\begin{equation}
\begin{split}
\pi^{3}(x_{1}x_{2}x_{3}) &= \sum_{j=0}^{3} P(x_{1}x_{2}x_{3} | j) \pi^{2}(j) \\
&= P(x_{1}x_{2}x_{3} | 00) \pi^{2}(00) + P(x_{1}x_{2}x_{3} | 01) \pi^{2}(01) \\
& \hspace{1cm} + P(x_{1}x_{2}x_{3} | 10) \pi^{2}(10) + P(x_{1}x_{2}x_{3} | 11) \pi^{2}(11),
\end{split}
\end{equation}
This can further be expressed in terms of the transition probabilities. For example, letting $x_{1}=1, x_{2}=0, x_{3}=1$ produces the following:
\begin{equation}
\pi^{3}(101) = \pi^{2}(00) p_{00}^{01} p_{01}^{10} p_{10}^{01} + \pi^{2}(01) p_{01}^{11} p_{11}^{10} p_{10}^{01} + \pi^{2}(10) p_{10}^{01} p_{01}^{10} p_{10}^{01} + \pi^{2}(11) p_{11}^{11} p_{11}^{10} p_{10}^{01}.
\end{equation}

\section{Inference} \label{Inf}

Given an observed sequence, $x$, of binary letters with an unknown de Bruijn correlation structure, inference can be performed by estimating both the de Bruijn word length $m$ and the corresponding transition probabilities $p_{i}^{j}$ to an associated de Bruijn process. The joint likelihood of the sequence follows from the correlated Bernoulli distribution in Theorem \ref{CBD1}.  Due to the different number of transition probabilities required for each word length, it is initially assumed that the word length is known. The likelihood of a given sequence $X$ for $m=2$ is given in Lemma \ref{TL2} and the likelihood for the general case where $m \ge 1$ is given in Theorem \ref{TLM}. In both cases, to simplify notation, $n_{i}^{j}$ denotes the number of times the transition from the word $i$ to the word $j$ takes place in the sequence. The numerical representation of the binary notation is also included in the general case. The likelihood can then be applied to estimate the transition probabilities through frequentist maximum likelihood estimation or using Bayesian methods. 

It is not obvious that the end form of the likelihood gives the joint probability of the sequence. The ordering of the letters in the sequence is fixed and, due to the de Bruijn structure, for each word in the sequence there are only two possible words that can follow. Hence, the distinct number of times each transition occurs in the sequence defines the exact ordering of the letters. 

\begin{lemma}[Transition Likelihood, $m=2$] \label{TL2}
\rm{For a given sequence, $x = \{x_{1}, x_{2}, \ldots, x_{n}\}$, the joint likelihood in terms of the length $m=2$ de Bruijn transition probabilities $p_{i}^{j}$ is given as follows:}
\begin{equation}
\begin{split}
\mathcal{L}(X|p) &= (p_{00}^{00})^{\sum_{i=1}^{n-2} (1-x_{i})(1-x_{i+1})(1-x_{i+2})} \times (p_{00}^{01})^{\sum_{i=1}^{n-2} (1-x_{i})(1-x_{i+1})(x_{i+2})} \\
& \hspace{1cm} \times (p_{01}^{10})^{\sum_{i=1}^{n-2} (1-x_{i})(x_{i+1})(1-x_{i+2})} \times (p_{01}^{11})^{\sum_{i=1}^{n-2} (1-x_{i})(x_{i+1})(x_{i+2})} \\
& \hspace{1cm} \times (p_{10}^{00})^{\sum_{i=1}^{n-2} (x_{i})(1-x_{i+1})(1-x_{i+2})} \times (p_{10}^{11})^{\sum_{i=1}^{n-2} (x_{i})(1-x_{i+1})(x_{i+2})} \\
& \hspace{1cm} \times (p_{11}^{10})^{\sum_{i=1}^{n-2} (x_{i})(x_{i+1})(1-x_{i+2})} \times (p_{11}^{11})^{\sum_{i=1}^{n-2} (x_{i})(x_{i+1})(x_{i+2})} \\
&= (p_{00}^{00})^{n_{00}^{00}} \hspace{0.1cm} (p_{00}^{01})^{n_{00}^{01}} \hspace{0.1cm} (p_{01}^{10})^{n_{01}^{10}} \hspace{0.1cm} (p_{01}^{11})^{n_{01}^{11}} \hspace{0.1cm} (p_{10}^{00})^{n_{10}^{00}} \hspace{0.1cm} (p_{10}^{01})^{n_{10}^{01}} \hspace{0.1cm} (p_{11}^{10})^{n_{11}^{10}} \hspace{0.1cm} (p_{11}^{11})^{n_{11}^{11}} \\
&= (1-p_{00}^{01})^{n_{00}^{00}} \hspace{0.1cm} (p_{00}^{01})^{n_{00}^{01}} \hspace{0.1cm} (1-p_{01}^{11})^{n_{01}^{10}} \hspace{0.1cm} (p_{01}^{11})^{n_{01}^{11}} \hspace{0.1cm} (1-p_{10}^{01})^{n_{10}^{00}} \hspace{0.1cm} (p_{10}^{01})^{n_{10}^{01}} \hspace{0.1cm} \\
& \hspace{1cm} \times (1-p_{11}^{11})^{n_{11}^{10}} \hspace{0.1cm} (p_{11}^{11})^{n_{11}^{11}},
\end{split}
\end{equation}
\end{lemma}

\begin{proof}
See Appendix
\end{proof}

\begin{theorem}[Transition Likelihood, $m \ge 1$] \label{TLM}
\rm{For a given sequence, $x = \{x_{1}, x_{2}, \ldots, x_{n}\}$, the joint likelihood in terms of the length $m$ de Bruijn transition probabilities $p_{i}^{j}$ is given as follows:}
\begin{equation}
\begin{split}
\mathcal{L}(X|p) &= \left( p_{0 \ldots 0}^{0 \ldots 0} \right)^{\sum_{i=1}^{n-m} (1-x_{i}) \ldots (1-x_{i+m})} \times \left( p_{0 \ldots 00}^{0 \ldots 01} \right)^{\sum_{i=1}^{n-m} (1-x_{i}) \ldots (1-x_{i+m-1})(x_{i+m})} \\
& \hspace{1cm} \times \ldots \times \left( p_{1 \ldots 11}^{1 \ldots 10} \right)^{\sum_{i=1}^{n-m} (x_{i}) \ldots (x_{i+m-1})(1-x_{i+m})} \times \left( p_{1 \ldots 1}^{1 \ldots 1} \right)^{\sum_{i=1}^{n-m} (x_{i}) \ldots (x_{i+m})} \\
&= \left( p_{0 \ldots 0}^{0 \ldots 0} \right)^{n_{0 \ldots 0}^{0 \ldots 0}} \times \left( p_{0 \ldots 00}^{0 \ldots 01} \right)^{n_{0 \ldots 00}^{0 \ldots 01}} \times \ldots \times \left( p_{1 \ldots 11}^{1 \ldots 10} \right)^{n_{1 \ldots 11}^{1 \ldots 10}} \times \left( p_{1 \ldots 1}^{1 \ldots 1} \right)^{n_{1 \ldots 1}^{1 \ldots 1}} \\
&= \prod_{i=0}^{2^{m+1}-1} \left(p_{\frac{1}{2} (i - (i \hspace{0.1cm} \text{mod } 2))}^{i \hspace{0.1cm} \text{mod } 2^{m}} \right) ^{n_{\frac{1}{2} (i - (i \hspace{0.1cm} \text{mod } 2))}^{i \hspace{0.1cm} \text{mod } 2^{m}}}\\
&= \prod_{i=0}^{2^{m}-1} \left( 1 - p_{i}^{(2i+1) \hspace{0.1cm} \text{mod } 2^{m}} \right)^{n_{i}^{((2i+1) \hspace{0.1cm} \text{mod } 2^{m}) - 1}} \left( p_{i}^{(2i+1) \hspace{0.1cm} \text{mod } 2^{m}} \right) ^{n_{i}^{((2i+1) \hspace{0.1cm} \text{mod } 2^{m})}}.
\end{split}
\end{equation}
\end{theorem}

\begin{proof}
See Appendix
\end{proof}

Given a maximum likelihood approach, uncertainties on the estimates of the transition probabilities can be calculated using the Fisher information \citep{Feller1950}, $I(p_{i}^{j}) = -\mathbb{E} \left[ \frac{\partial^{2} \text{log}(\mathcal{L})}{\partial {p_{i}^{j}}^{2}} \right]$. Note that $p_{i}^{j}$ still denotes the probability of transitioning from the word $i$ to the word $j$, and $\text{log} (\mathcal{L})$ is the natural log of the likelihood given in Theorem \ref{TLM}. The expectation of a function $g(x_{1}, \ldots, x_{n})$ with respect to the $x_{i}$ is given as:
\begin{equation}
\mathbb{E}[g(X_{1}, \ldots, X_{n})] = \sum_{x_{1}=0}^{x_{1}=1} \cdots \sum_{x_{n}=0}^{x_{n}=1} g(x_{1}, \ldots, x_{n}) \pi(x_{1}, \ldots, x_{n}),
\end{equation}
where $\pi(x_{1}, \ldots, x_{n})$ is taken to be the correlated Bernoulli distribution from Theorem \ref{CBD1}. The expectation is hence dependent on the length of the sequence$n$, word length and transition probabilities.  Following from this, the Fisher information for the general case for word length $m$ is given in Theorem \ref{FI}.

\begin{theorem}[Fisher information, $m \ge 1$] \label{FI}
\rm{The Fisher information, } $I(p_{k}) = -\mathbb{E}\left[\frac{\partial^{2}\text{log}(\mathcal{L})}{\partial {p_{k}}^{2}}\right]$, \rm{ for each transition probability $p_{k} = p_{\frac{1}{2}(k - (k \text{ mod} 2))}^{k \text{ mod} 2^{m}}$ for $k=0, 1, \ldots, m-1$, sequence length $n$ and word length $m$ is given by:}
\begin{equation}
\begin{split}
I(p_{k}) &= -\mathbb{E}\left[\frac{\partial^{2}\text{log}(\mathcal{L})}{\partial {p_{k}}^{2}}\right] \\
&= \frac{1}{p_{k}} \sum_{i=0}^{n-m-1} \sum_{j=0}^{2^{n-m-1}-1} \pi^{n} \left( 2^{m+1} j + 2^{i}k - \left( 2^{m+1} - 1 \right) \left( j \hspace{0.1cm} \text{mod} 2^{i} \right) \right)
\end{split}
\end{equation}
\end{theorem}

\begin{proof}
See Appendix
\end{proof}

Consider the case for the transition probability $p_{00}^{00}$ where $m=2$ and $n=3$. The expectation of the double derivative with respect to $p_{00}^{00}$ is as follows:
\begin{equation}
\begin{split}
\mathbb{E}\left[\frac{\partial^{2}\text{log}(\mathcal{L})}{\partial {p_{00}^{00}}^{2}}\right] &= -\frac{1}{\left(p_{00}^{00}\right)^{2}} \mathbb{E}\left[ (1-x_{1})(1-x_{2})(1-x_{3}) \right] \\
&= -\frac{1}{\left(p_{00}^{00}\right)^{2}} \sum_{x_{1}=0}^{x_{1}=1} \sum_{x_{2}=0}^{x_{2}=1} \sum_{x_{3}=0}^{x_{3}=1} (1-x_{1})(1-x_{2})(1-x_{3}) \pi^{(3)}(x_{1}, x_{2}, x_{3}) \\
&= -\frac{1}{\left(p_{00}^{00}\right)^{2}} \pi^{3}(000)
\end{split}
\end{equation}
Expanding this for general $n \ge 3$, results in:
\begin{equation}
\begin{split}
\mathbb{E}\left[\frac{\partial^{2}\text{log}(\mathcal{L})}{\partial {p_{00}^{00}}^{2}}\right] &= -\frac{1}{p_{00}^{00}} \sum_{i=0}^{n-3} \sum_{j=0}^{2^{n-3}-1} \pi^{n} \left( 2^{3} j - \left( 2^{3} - 1 \right) \left( j \hspace{0.1cm} \text{mod} 2^{i} \right) \right) \\
&= -\frac{1}{p_{00}^{00}} \sum_{i=0}^{n-3} \sum_{j=0}^{2^{n-3}-1} \pi^{n} \left( 8 j - 7 \left( j \hspace{0.1cm} \text{mod} 2^{i} \right) \right)
\end{split}
\end{equation}
Then for any transition probability, $p$, the general result for $m=2$ is as follows:
\begin{equation}
\begin{split}
\mathbb{E}\left[\frac{\partial^{2}\text{log}(\mathcal{L})}{\partial {p_{k}}^{2}}\right] &= -\frac{1}{p_{k}} \sum_{i=0}^{n-3} \sum_{j=0}^{2^{n-3}-1} \pi^{n} \left( 2^{3} j + 2^{i}k - \left( 2^{3} - 1 \right) \left( j \hspace{0.1cm} \text{mod} 2^{i} \right) \right) \\
&= -\frac{1}{p_{k}} \sum_{i=0}^{n-3} \sum_{j=0}^{2^{n-3}-1} \pi^{n} \left( 8 j + 2^{i}k - 7 \left( j \hspace{0.1cm} \text{mod} 2^{i} \right) \right),
\end{split}
\end{equation}
for $\{p_{0}, p_{1}, p_{2}, p_{3}, p_{4}, p_{5}, p_{6}, p_{7}\} = \{p_{00}^{00}, p_{00}^{01}, p_{01}^{10}, p_{01}^{11}, p_{10}^{00}, p_{11}^{10}, p_{11}^{11}\}$ and $n \ge 3$. The fisher information is then given by $- \mathbb{E} \left[\frac{\partial^{2}\text{log}(\mathcal{L})}{\partial {p_{k}}^{2}}\right]$.

Alternatively, the likelihood in Theorem \ref{TLM} can be used to estimate the transition probabilities using Bayes' theorem. The advantage to this is that the prior distribution can be used to specify any prior knowledge already known about the transition probabilities. For example, it may be known that the sequence is very clustered in blocks of $1$'s, and this can be incorporated into the prior distribution to put higher weighting onto the transition that remains at the all $1$ word. Due to the form of the likelihood (Theorem \ref{TLM}),  a Beta prior of the form, $P(p) = \frac{\Gamma(\alpha + \beta)}{\Gamma(\alpha)\Gamma(\beta)}p^{\alpha -1}(1-p)^{\beta -1}$, for $\alpha > 0$ and $\beta > 0$,  is used to produce the posterior distribution for the transition probabilities. 

First consider the de Bruijn word length $m=2$ case. The transition likelihood in Lemma \ref{TL2} can be combined with a Beta prior to obtain the posterior distribution as follows:
\begin{equation}
\begin{split}
P(p|X) &= \frac{\mathcal{L}(X|p)P(p)}{\int \mathcal{L}(X|p)P(p) dp} \\
&\propto \hspace{0.2cm}  \mathcal{L}(X|p)P(p) \\
&= \hspace{0.2cm} (1 - p_{00}^{01})^{n_{00}^{00}+\beta_{1}-1} \hspace{0.1cm} (p_{00}^{01})^{n_{00}^{01}+\alpha_{1}-1} \hspace{0.2cm}  \\
& \hspace{0.8cm} \times (1 - p_{01}^{11})^{n_{01}^{10}+\beta_{2}-1} \hspace{0.1cm} (p_{01}^{11})^{n_{01}^{11}+\alpha_{2}-1} \hspace{0.2cm}  \\
& \hspace{0.8cm} \times (1 - p_{10}^{01})^{n_{10}^{00}+\beta_{3}-1} \hspace{0.1cm} (p_{10}^{01})^{n_{10}^{01}+\alpha_{3}-1} \hspace{0.2cm}  \\ 
& \hspace{0.8cm} \times (1 - p_{11}^{11})^{n_{11}^{10}+\beta_{4}-1} \hspace{0.1cm} (p_{11}^{11})^{n_{11}^{11}+\alpha_{4}-1}.
\end{split}
\end{equation}

Clearly, the posterior distribution for the de Bruijn process transition probabilities is a product of beta densities.  Since the prior and posterior distributions are conjugate, the following can be stated:
\begin{equation} \label{MEF2}
\begin{split}
\int \mathcal{L}(X|p)P(p) dp = & \frac{\Gamma(n_{00}^{00}+\beta_{1})\Gamma(n_{00}^{01}+\alpha_{1})}{\Gamma(n_{00}^{00}+n_{00}^{01}+\beta_{1}+\alpha_{1})} \times \frac{\Gamma(n_{01}^{10}+\beta_{2})\Gamma(n_{01}^{11}+\alpha_{2})}{\Gamma(n_{01}^{10}+n_{01}^{11}+\beta_{2}+\alpha_{2})} \times \\
& \frac{\Gamma(n_{10}^{00}+\beta_{3})\Gamma(n_{10}^{01}+\alpha_{3})}{\Gamma(n_{10}^{00}+n_{10}^{01}+\beta_{3}+\alpha_{3})} \times 
\frac{\Gamma(n_{11}^{10}+\beta_{4})\Gamma(n_{11}^{11}+\alpha_{4})}{\Gamma(n_{00}^{10}+n_{11}^{11}+\beta_{4}+\alpha_{4})},
\end{split}
\end{equation}
which specifies the model evidence. For the general case when $m \ge 1$, the posterior distribution for the transition probabilities is given in Theorem \ref{MEM}.

\begin{theorem}[Posterior Distribution for de Bruijn Probability Transitions, $m \ge 1$] \label{MEM}
\rm{Applying Bayes' theorem, the posterior distribution of the de Bruijn transition probabilities is:}
\begin{equation}
\begin{split}
P(p | X) & = \frac{\mathcal{L}(X|p,m) P(p|m)}{P(X)} \\
& = \frac{\mathcal{L}(X|p,m)P(p|m)}{\int \mathcal{L}(X|p,m)P(p|m) dp}
\end{split}
\end{equation}
\rm{where,}
\begin{equation}
\begin{split}
\mathcal{L}(X|p,m)P(p|m) &= \prod_{i=0}^{2^{m}-1} (1 - p_{i}^{((2i+1) \hspace{0.1cm} \text{mod } 2^{m})})^{n_{i}^{((2i+1) \hspace{0.1cm} \text{mod } 2^{m})-1 }+\beta_{i+1}-1} \\
& \hspace{2cm} \times (p_{i}^{((2i+1) \hspace{0.1cm} \text{mod } 2^{m})})^{n_{i}^{((2i+1) \hspace{0.1cm} \text{mod } 2^{m}) }+\alpha_{i+1}-1}
\end{split}
\end{equation}
\rm{and}
\begin{equation}
\int P(X|p,m)P(p|m) dp = \prod_{i=0}^{2^{m}-1} \frac{\Gamma(n_{i}^{((2i+1) \hspace{0.1cm} \text{mod } 2^{m}) - 1} + \beta_{i+1})\Gamma(n_{i}^{((2i+1) \hspace{0.1cm} \text{mod } 2^{m})}) + \alpha_{i+1})}{\Gamma(n_{i}^{((2i+1) \hspace{0.1cm} \text{mod } 2^{m}) - 1} + n_{i}^{((2i+1) \hspace{0.1cm} \text{mod } 2^{m})} +\beta_{i+1}+\alpha_{i+1})}
\end{equation}
\end{theorem}

\begin{proof}
See Appendix
\end{proof}

Estimating the de Bruijn word length $m$ which was most likely used to generate the sequence is a more challenging problem since a different number of transition probabilities are required for different word lengths. Since the word lengths can only take integer values, we have chosen to proceed by comparing models using Bayes' factors \citep{Kass1995, OHagan1997}.  The method of Bayes' factors considers whether an observed sequence $X$ of binary random variables was generated from either a word length $m_{1}$ de Bruijn process (hypothesis 1) with probability $P(X | m_{1})$, or from a length $m_{2}$ de Bruijn process (hypothesis 2) with probability $P(X | m_{2})$.  The prior probabilities $P(m_{1})$ and $P(m_{2})$ are also defined,  giving the probability that the sequence was indeed generated using a length $m_{1}$ or $m_{2}$ de Bruijn process respectively.  When combined with the data, this then gives appropriate posterior probabilities $P(m_{1}| X)$ and $P(m_{2}| X) = 1 - P(m_{1}| X)$. If Bayes' theorem is considered in terms of an odds scale of these hypotheses when in favour of $m_{1}$,  we have the following relationship:
\begin{equation}
\frac{P(m_{1} | X)}{P(m_{2} | X)} = \frac{P(X | m_{1})}{P(X | m_{2})} \frac{P(m_{1})}{P(m_{2})}.
\end{equation} 
If we say that the hypotheses $m_{1}$ and $m_{2}$ are equally likely then the Bayes' factor can be defined to be the posterior odds in favour of $m_{1}$:
\begin{equation}
B_{1,2} = \frac{P(X | m_{1})}{P(X | m_{2})}.
\end{equation}

Since the transition probabilities are unknown parameters in this case,  an expression for $P(X | m_{k})$ can be found by integrating over the parameter space. This becomes:
\begin{equation}
P(X | m_{k}) = \int \mathcal{L}(X | p_{k}, m_{k}) P(p_{k} | m_{k}) \hspace{0.2cm} dp_{k},
\end{equation}
for $k \in \{1,2\}$, where $\mathcal{L}(X | p_{k}, m_{k})$ is the likelihood of the data and $P(p_{k} | m_{k})$ is the prior density of the model parameters, $p$. There is an obvious similarity between this expression and the model evidence in Theorem \ref{MEM}. Due to the conjugate priors,  the Bayes' factor ratio is stated to be the ratio of the model evidences for each of the model hypotheses, which is shown in Theorem \ref{MEM2}. We note here that it is not necessary to calculate the posterior on the transition probabilities since the expression for the Bayes' factor is only dependent on the prior density. For each calculation of $P(X | m_{k})$, we will know the length $m_{k}$ and hence the quantity of parameters, $p$, which are to be estimated. 

In the set up of the de Bruijn process, we make the assumption that the word length, $m$, will remain fairly small. This is considered reasonable since large word lengths create a vast number of transition probabilities to be estimated, and the increase in dimension does not have much effect on the accuracy of the estimates. Therefore, we make the choice to limit word lengths to not be greater than $10$.  This is a pragmatic limit, and could be increased if required. Hence, to choose the word length that best represents the data,  calculate $B_{i,j} = \frac{P(X | m_{i})}{P(X | m_{j})}$ for each pair of models where $i,j \in \{1,2,...,10\}$ and select the value for $m$ in which the Bayes' factor is consistently higher. When values of $B_{i,j}$ are large, this gives more evidence to reject the model with word length $m_{i}$ in favour of the model with word length $m_{j}$. 

By only selecting 10 different models to compare and choosing the one that best represents the data, we acknowledge that we have implemented a frequentist aspect to our method. Instead, we could opt to do this in a fully Bayesian way to maximise the Bayes' factor and allow any word length to be considered. However, to do this we would have to put a fairly strong prior on $m$ to minimise large potential word lengths. This is left for future work.

\begin{theorem}[Estimation of De Bruijn Word length by Bayes' factors, $m \ge 1$] \label{MEM2}
\
\rm{Consider a sequence of 0's and 1's which was generated under one of two hypotheses. The first is a de Bruijn process with word length $m_{1}$ and the second is a de Bruijn process with word length $m_{2}$. The Bayes' factor ratio is as follows:}
\begin{equation}
B_{1,2} = \frac{P(X | m_{1})}{P(X | m_{2})}
\end{equation}
\rm{where,}
\begin{equation}
\begin{split}
P(X | m_{k}) & = \int P(X|p,m_{k})P(p|m_{k}) dp \\
& = \prod_{i=0}^{2^{m_{k}}-1} \frac{\Gamma(n_{i}^{((2i+1) \hspace{0.1cm} \text{mod } 2^{m_{k}}) - 1} + \beta_{i+1})\Gamma(n_{i}^{((2i+1) \hspace{0.1cm} \text{mod } 2^{m_{k}})}) + \alpha_{i+1})}{\Gamma(n_{i}^{((2i+1) \hspace{0.1cm} \text{mod } 2^{m_{k}}) - 1} + n_{i}^{((2i+1) \hspace{0.1cm} \text{mod } 2^{m_{k}})} +\beta_{i+1}+\alpha_{i+1})} ,
\end{split}
\end{equation}
\rm{for} $k \in \{1,2\}$.
\rm{When values of} $B_{1,2}$ \rm{are large, we have more evidence to reject the first hypothesis with word length} $m_{1}$ \rm{in favour of the second hypothesis with word length} $m_{2}$. 
\end{theorem}

\begin{proof}
Follows from Theorem \ref{MEM}
\end{proof}

\subsection{Examples}

The following examples illustrate the method of inference so that given a sequence of binary random variables,  the most likely de Bruijn process can be estimated. This includes both the word length $m$, and the transition probabilities, $p_{i}^{j}$. 

The first sequence is an anti-clustered alternating sequence given by the top panel in Figure \ref{Samp2I}. This sequence is of length $n=200$ and was generated using an $m=2$ de Bruijn process with transition probabilities: $\{ p_{00}^{01}, p_{01}^{11}, p_{10}^{01}, p_{11}^{11} \} = \{ 0.9, 0.25, 0.75, 0.1\}$. We begin by trying to estimate the word length using the methods of Bayes' factors with the model evidence stated in Theorem \ref{MEM2}. The model evidence for the sequence is calculated for each possible de Bruijn process with word lengths, $m = 1, ..., 10$, for comparison.  We do not assume any prior knowledge, so let each $\alpha = \beta = 1$ for the equivalence of a uniform prior. After calculating the Bayes' factor ratio for each pair of proposed models, we can conclude that the sequence was most likely generated with a length $m=2$ de Bruijn process.

\begin{figure}[ht]
\centering
\includegraphics[scale=0.5]{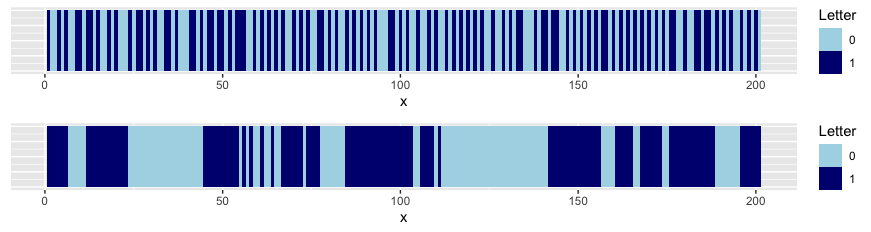}
\caption{Samples from length $m=2$ (top) and length $m=3$ (bottom) de Bruijn processes with letters 0 and 1. The transition probabilities are: $\{ p_{00}^{01}, p_{01}^{11}, p_{10}^{01}, p_{11}^{11} \} = \{ 0.9, 0.25, 0.75, 0.1 \}$ and $\{p_{000}^{001}, p_{001}^{011}, p_{010}^{101}, p_{011}^{111}, p_{100}^{001}, p_{101}^{011}, p_{110}^{101}, p_{111}^{111} \} = \{0.1, 0.7, 0.5, 0.8, 0.2, 0.5, 0.3, 0.9 \}$ respectively.}
\label{Samp2I}
\end{figure}

The left plot in Figure \ref{Hist2I} shows a histogram of the estimated word lengths for the $m=2$ de Bruijn process. We generated $1000$ sequences from the same de Bruijn process used to generate the top sequence in Figure \ref{Samp2I}, and used our method of inference to estimate the word length used for each sequence. The histogram shows that nearly every sequence was estimated to be generated from a length $m=2$ de Bruijn process. 

\begin{figure}[ht]
\centering
\includegraphics[scale=0.45]{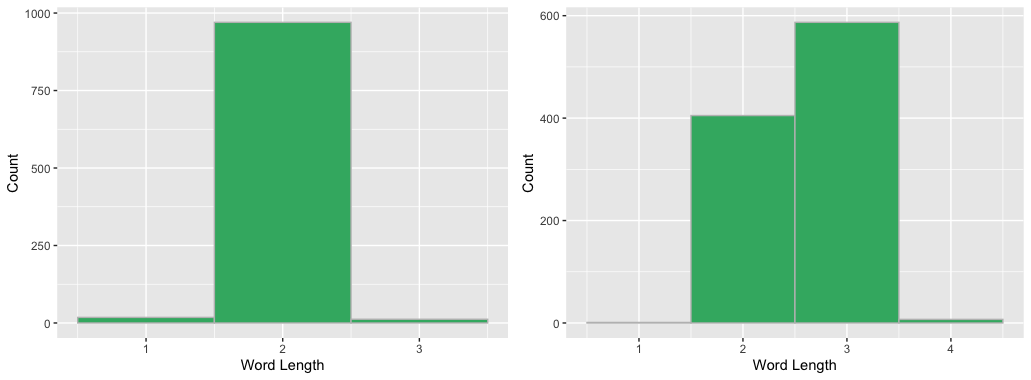}
\caption{Histograms of estimated word lengths from 1000 sequences generated from the $m=2$ (left) and $m=3$ (right) examples in Figure \ref{Samp2I}.}
\label{Hist2I}
\end{figure}

Finally, given the word length was estimated to be $m=2$, we then estimated the given transition probabilities using a simple Metropolis Hastings MCMC approach. Using the likelihood of the sequence given in Lemma \ref{TL2} and non-informative priors, the estimated parameters are given in Table \ref{TEInf} (left) along with $95\%$ confidence intervals. All estimates are close to the true values, where the true values lie within the confidence intervals for all parameters. On average, the true values and the expected values differ by $0.018$ indicating a small level of uncertainty.

The estimates of the parameters improve greatly with the increased length of the sequence. Table \ref{TEInf2} shows the effects of altering the length of the sequence. Using the same de Bruijn process with transition probabilities $\{p_{00}^{00}, p_{01}^{11}, p_{10}^{01}, p_{11}^{11}\} = \{0.9, 0.25, 0.75, 0.1\}$, we simulated $100$ sequences for lengths $n=50$, $n=100$, $n=200$ and $n=500$. The transition probabilities are estimated for each of these sequences and Table \ref{TEInf2} gives the average estimate along with a $95\%$ interval. Although the average estimate does not change much, it is clear that the interval of possible values reduces as the length of the sequence increases.

\begin{table}[h!]
\centering
\begin{tabular}{||c | c c ||} 
\hline
$p_{i}^{j}$ & True value & Estimate [$95\%$ CI] \\ [0.5ex]
\hline\hline
$p_{00}^{01}$ & $0.9$ & $0.932 \hspace{0.3cm} [0.883,0.969]$ \\ 
$p_{01}^{11}$ & $0.25$ & $0.245 \hspace{0.3cm} [0.206,0.286]$ \\
$p_{10}^{01}$ & $0.75$ & $0.737 \hspace{0.3cm} [0.696,0.775]$ \\
$p_{11}^{11}$ & $0.1$ & $0.079 \hspace{0.3cm} [0.024,0.136]$ \\
& & \\
& & \\
& & \\
& & \\ [1ex] 
\hline
\end{tabular}
\quad
\begin{tabular}{||c | c c ||} 
\hline
$p_{i}^{j}$ & True value & Estimate [$95\%$ CI] \\ [0.5ex]
\hline\hline
$p_{000}^{001}$ & $0.1$ & $0.132 \hspace{0.3cm} [0.091,0.166]$ \\ 
$p_{001}^{011}$ & $0.7$ & $0.691 \hspace{0.3cm} [0.582,0.793]$ \\
$p_{010}^{101}$ & $0.5$ & $0.516 \hspace{0.3cm} [0.383,0.650]$ \\
$p_{011}^{111}$ & $0.8$ & $0.737 \hspace{0.3cm} [0.683,0.844]$ \\
$p_{100}^{001}$ & $0.2$ & $0.228 \hspace{0.3cm} [0.139,0.323]$ \\ 
$p_{101}^{011}$ & $0.5$ & $0.504 \hspace{0.3cm} [0.382,0.589]$ \\
$p_{110}^{101}$ & $0.3$ & $0.393 \hspace{0.3cm} [0.276,0.499]$ \\
$p_{111}^{111}$ & $0.9$ & $0.894 \hspace{0.3cm} [0.868,0.918]$ \\ [1ex] 
\hline
\end{tabular}
\caption{Tables to show the estimates of the transition probabilities for the sequences in Figure \ref{Samp2I} with transition probabilities $\{ p_{00}^{01}, p_{01}^{11}, p_{10}^{01}, p_{11}^{11} \} = \{ 0.9, 0.25, 0.75, 0.1 \}$ (left) and $\{p_{000}^{001}, p_{001}^{011}, p_{010}^{101}, p_{011}^{111}, p_{100}^{001}, p_{101}^{011}, p_{110}^{101}, p_{111}^{111} \} = \{0.1, 0.7, 0.5, 0.8, 0.2, 0.5, 0.3, 0.9 \}$ (right). The true value is given along with the estimate and $95\%$ confidence interval.}
\label{TEInf}
\end{table}

\begin{table}[h!]
\centering
\begin{tabular}{||c | c c c c ||} 
\hline
$p_{i}^{j}$ & $n=50$ & $n=100$ & $n=200$ & $n=500$ \\ [0.5ex] 
\hline\hline
$p_{00}^{01}=0.9$ & $0.905 \hspace{0.3cm} [0.687,0.999]$ & $0.902 \hspace{0.3cm} [0.768,0.999]$ & $0.902 \hspace{0.3cm} [0.827,0.998]$ & $0.901 \hspace{0.3cm} [0.854,0.954]$ \\ 
$p_{01}^{11}=0.25$ & $0.262 \hspace{0.3cm} [0.098,0.424]$ & $0.249 \hspace{0.3cm} [0.149,0.364]$ & $0.251 \hspace{0.3cm} [0.176,0.315]$ & $0.249 \hspace{0.3cm} [0.203,0.287]$ \\
$p_{10}^{01}=0.75$ & $0.745 \hspace{0.3cm} [0.553,0.905]$ & $0.753 \hspace{0.3cm} [0.641,0.848]$ & $0.757 \hspace{0.3cm} [0.697,0.813]$ & $0.754 \hspace{0.3cm} [0.704,0.787]$ \\
$p_{11}^{11}=0.1$ & $0.082 \hspace{0.3cm} [0.001,0.287]$ & $0.082 \hspace{0.3cm} [0.001,0.222]$ & $0.080 \hspace{0.3cm} [0.001,0.178]$ & $0.102 \hspace{0.3cm} [0.057,0.156]$ \\ [1ex] 
\hline
\end{tabular}
\caption{Table showing the effects of altering the lengths of sequences on estimating transition probabilities. For each length ($n=50$, $n=100$, $n=200$, $n=500$), 100 sequences are generated from the $m=2$ de Bruijn process with transition probabilities $\{ p_{00}^{01}, p_{01}^{11}, p_{10}^{01}, p_{11}^{11} \} = \{ 0.9, 0.25, 0.75, 0.1 \}$ and the parameters are estimated. The average estimates along with $95\%$ intervals are given.}
\label{TEInf2}
\end{table}

The second example (bottom panel in Figure \ref{Samp2I}) is a length $n=200$ sequence designed such that there are large clustering blocks of $0$'s and $1$'s, with independent Bernoulli patterns occurring occasionally. It is generated using a length $m=3$ de Bruijn process with transition probabilities: $\{ p_{000}^{001}, p_{001}^{011}, p_{010}^{101}, p_{011}^{111}, p_{100}^{001}, p_{101}^{011}, p_{110}^{101}, p_{111}^{111} \} = \{ 0.1, 0.7, 0.5, 0.8, 0.2, 0.5, 0.3, 0.9 \}$.

As for the previous example, we begin by estimating the de Bruijn word most likely used to generate the sequence using Bayes' factors. This was estimated to be $m=3$. Although estimated correctly, by observing the right histogram in Figure \ref{Hist2I}, we can see that this isn't always the case for sequences generated from this de Bruijn process. After simulating $1000$ sequences and estimating the word length for each, almost $60\%$ are estimated to be $m=3$, but just over $40\%$ are actually estimated to be $m=2$. This is because some of the sequences generated have similar correlation structures to sequences generated with an $m=2$ de Bruijn process. The extra transition probabilities from the $m=3$ de Bruijn process have little effect and it is likely that these sequences could have been generated with transition probabilities close to $\{ p_{00}^{01}, p_{01}^{11}, p_{10}^{01}, p_{11}^{11} \} = \{ 0.1, 0.65, 0.35, 0.9\}$.

Finally, the transition probabilities (given $m=3$) are estimated and given in Table \ref{TEInf} (right) along with $95\%$ confidence intervals. All estimates are reasonably accurate, where on average, the estimates and true value differ by $0.020$. All true values lie within the $95\%$ confidence intervals.

\section{Application: The Boat Race} \label{App}

As a real world example, we will be applying the de Bruijn process inference to the annual boat race between Cambridge university and Oxford university. The race was first held in 1829 and has continued every year until present with a few missing years (partly due to the first and second World Wars and COVID-19). There are 166 data points in total where Oxford has won the race 80 times, Cambridge has won 85 times, and they have drawn once (1877). We also note that two races occurred in 1849, and so for simplicity reasons we chose to ignore the second recorded race (where Oxford won) and the race where the two universities drew.

To model the data using the de Bruijn process, let $0$ represent the years that Oxford won the race and let $1$ represent the years that Cambridge won the race. Hence $\pi(\rm{Oxford}) = \pi(0) = 79/164 = 0.482$ and $\pi(\rm{Cambridge}) = \pi(1) = 85/164 = 0.518$. The data is shown in Figure \ref{boat1} (where dark blue represents Oxford and light blue represents Cambridge).

\begin{figure}[ht]
\centering
\includegraphics[scale=0.25]{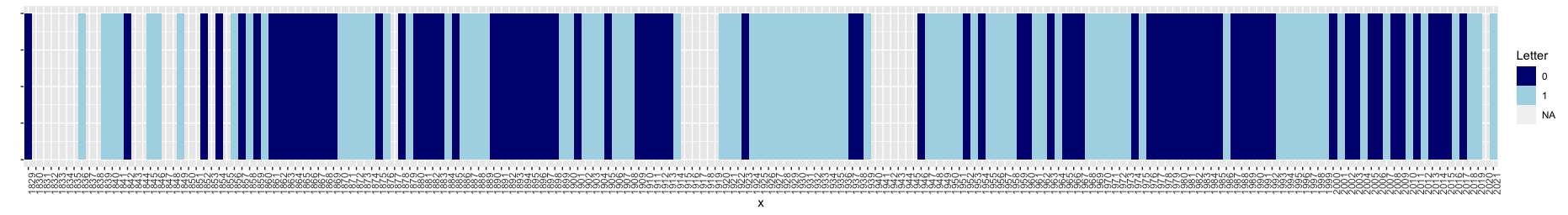}
\caption{Data for the annual boat race between Cambridge and Oxford universities. Races where Oxford have won are represented by a $0$ and shown in dark blue. Races where Cambridge have won are represented by a $1$ and are shown in light blue.}
\label{boat1}
\end{figure}

As before, we begin by estimating the word length for the estimated de Bruijn process. This is done using the Bayes' factor method explained in Section \ref{Inf} and is best estimated to be $m=2$. An explanation for dealing with the missing values is given in the Appendix. Since the Bayes' factor method compares each pair of models in turn, we noticed that there was little difference between the $m=2$ and $m=3$ models, hence $m=3$ would still likely give valid results. Given both of these possible word lengths, the associated transition probabilities are estimated using Bayesian methods. These are given in Table \ref{btable2} along with $95\%$ credible intervals.

\begin{table}[h!]
\centering
\begin{tabular}{||c | c c ||} 
\hline
$p_{i}^{j}$ & Estimate & $95\%$ CI \\ [0.5ex]
\hline\hline
$p_{00}^{01}$ & $0.283$ & $[0.175, 0.386]$ \\ 
$p_{01}^{11}$ & $0.462$ & $[0.327, 0.601]$ \\
$p_{10}^{01}$ & $0.519$ & $[0.372, 0.652]$ \\
$p_{11}^{11}$ & $0.723$ & $[0.618, 0.813]$ \\
& & \\
& & \\
& & \\
& & \\ [1ex] 
\hline
\end{tabular}
\quad
\begin{tabular}{||c | c c ||} 
\hline
$p_{i}^{j}$ & Estimate & $95\%$ CI \\ [0.5ex]
\hline\hline
$p_{000}^{001}$ & $0.244$ & $[0.142, 0.364]$ \\ 
$p_{001}^{011}$ & $0.458$ & $[0.237, 0.686]$ \\
$p_{010}^{101}$ & $0.362$ & $[0.175, 0.551]$ \\
$p_{011}^{111}$ & $0.817$ & $[0.625, 0.955]$ \\
$p_{100}^{001}$ & $0.340$ & $[0.188, 0.594]$ \\ 
$p_{101}^{011}$ & $0.467$ & $[0.251, 0.691]$ \\
$p_{110}^{101}$ & $0.671$ & $[0.444, 0.833]$ \\
$p_{111}^{111}$ & $0.680$ & $[0.549, 0.789]$ \\ [1ex] 
\hline
\end{tabular}
\caption{Table to show the estimated transition probabilities for the boat race data given in Figure \ref{boat1} when the word length is given to be either $m=2$ (left) or $m=3$ (right). Estimates are given alongside $95\%$ credible intervals.}
\label{btable2}
\end{table}

Initially considering the case when $m=2$, the marginal probabilities for the letters and words are given as $\pi(0) = 0.482$, $\pi(1) = 0.518$ and $\{\pi(00), \pi(01), \pi(10), \pi(11)\} = \{0.303, 0.191, 0.184, 0.322\}$. So far Cambridge have won more races than Oxford and there were also more times that Cambridge won consecutively. From the transition probability estimates, there is a $46.2\%$ chance for Cambridge to win again if they only won the previous race. This increases to a $72.3\%$ chance for Cambridge to win again if they have already won the past two races. Equivalently for Oxford, there is a $48.1\%$ chance for them to win again if they only won the previous race. But this again increases to a $71.7\%$ chance for Oxford to win again if they have already won the past two races.

We can hence conclude that both universities increase their chances of winning the next race if they have won previous races consecutively. This is slightly higher for Cambridge (increase of $0.6\%$ compared to Oxford when considering the two previous races) and is likely to be due to having overlapping team members in these years or having more confidence to win again if they have done previously.

If the sequence of races is $01$, then there is a $53.8\%$ chance for Oxford to win the next race. Alternatively if the sequence is $10$, then there is a $51.9\%$ chance for Cambridge to win the next race. There is almost a $50\%$ chance for either university to win the next race if the opposite university won the last race. Oxford have a slightly higher chance of winning, implying that they may have a larger drive to win again if Cambridge won the last race.
 
If we let $m=3$, the marginal probabilities for the letters and words are instead estimated to be $\pi(0) = 0.482$, $\pi(1) = 0.518$ and $\{\pi(000), \pi(001), \pi(010), \pi(011), \pi(100), \pi(101), \pi(110), \pi(111)\} = \{0.226, 0.089, 0.096, 0.082, 0.089, 0.096, 0.089, 0.233\}$.

Given $m=3$, there is now a $81.7\%$ chance for Cambridge to win the next race if they have won the last two races. This falls to a $68.0\%$ chance of winning if they have won the last three races. Hence, there is a much higher chance for Cambridge to win consecutive races, but the increased chance in winning reduces when the number of wins in a row exceeds two. Again, this is likely to be due to the team members taking part, winning tactics being used or previous wins providing increased confidence. The individual rowers are fairly likely to take part in the race for two years in a row, but less likely to take part in further races due to finishing their studies at the university. 

Oxford now only has a $66.0\%$ chance to win the next race if they have won the last two races, which now increases to a $75.6\%$ chance to win the race if they have won the last three races. Unlike Cambridge, Oxford increases their chance of winning the more races they win in row, indicating that they tend to value past tactics or confidence from winning over individual team members.

We observe that as well as having a higher chance of a longer streak of winning, Oxford also tends to perform better when neither university has built up a run in previous races. When the sequences of races is $101$, the probability of Oxford winning is $53.3\%$, whilst when the sequence is $010$, Cambridge has a $36.2\%$ chance of winning the next race. This is similar when we just consider the previous race. If Cambridge won the previous race then they only have a $46.3\%$ chance of winning the next race (calculated using law of total probability), compared with if Oxford won the previous race then Oxford have a $47.8\%$ chance of winning again.

We can also use the de Bruijn process to predict the results of the 2022 boat race. If $m=2$, then there is a $59.7\%$ ($p_{01}^{11} \pi(0) + p_{11}^{11} \pi(1) = 0.597$) chance of Cambridge winning the race and a $40.3\%$ of Oxford winning the race in 2022. If $m=3$, then there is a $57.7\%$ ($p_{101}^{011} \pi(0) + p_{111}^{111} \pi(1) = 0.577$) chance of Cambridge winning the race and a $42.3\%$ chance of Oxford winning the race in 2022.  Oxford won the race in 2022.

\section{Discussion} \label{FW}

The de Bruijn process was first introduced in \cite{Kimpton2022} where it was shown that correlated sequences of binary random variables could be modelled using a de Bruijn process.  In this paper, we have extended these ideas further to include both a formal definition for the joint distribution of many correlated Bernoulli trials using the de Bruijn process, as well as a method of inference. After establishing a method of inference, we applied our method to an application; the Oxford and Cambridge university boat race where we have used previous data to predict the results for the 2022 race.

Given a sequence of correlated binary data, we have shown that we can fit a de Bruijn process by estimating both the de Bruijn word length and associated transition probabilities. Since the word length can only take integer values, we have chosen to estimates the best word length using Bayes' factors.  We use Bayes' factors to test whether a sequence was most likely generated using a de Bruijn process with word length $m_{1}$ or a de Bruijn process with word length $m_{2}$. By comparing all possible models with words length $1, \ldots, 10$, we can then select the one that performs the best across all comparisons.  Once we have estimated the word length, we can then estimate the associated transition probabilities.  

The first obvious extension to the 1d de Bruijn process is to extend it to higher dimensions. This is a difficult problem due to the distinct direction of de Bruijn graphs which defines the correlation between past and future successes. For example, in two dimensions it is clear that it is not easy to enforce a direction onto a two dimensional grid.  Therefore, we also propose a non-directional de Bruijn process to remove any forced direction, but maintain the word structure. The de Bruijn word structure is important since it allows us to alter and control the spread of correlation, however, specifically in a 2d grid, letters are dependent on more letters in their neighbourhood than just a small selection in a set direction. Although this is a hard problem, we initially aim to change the form of the de Bruijn word further to see if the word structure can exist in a non-directional framework.

\bibliographystyle{abbrvnat}
\bibliography{DBP2bib2}

\newpage

\section*{Appendix}

\section*{Theorem 2.1 (Correlated Bernoulli Distribution $n \ge m$)}
For correlated binary random sequence, $X = \{X_{1}, X_{2}, \ldots, X_{n}\}$, where $x= \{x_{1}, x_{2}, \ldots, x_{n}\}$ is a realisation from $X$, the joint probability density is given as follows:
\begin{equation}
\begin{split}
\pi^{n}(x_{1}&, x_{2}, \ldots, x_{n}) \\
&= \pi^{n}(0...0)^{(1-x_{1})...(1-x_{n})} \times \pi^{n}(0...01)^{(1-x_{1})...(1-x_{n-1})(x_{n})} \times \ldots \\
&\hspace{1cm} \times \pi^{n}(1...10)^{(x_{1})...(x_{n-1})(1-x_{n})} \times \pi^{n}(1...1)^{(x_{1})...(x_{n})} \\
&= \prod_{i=0}^{2^{n}-1} \bigg( \pi^{n}(i) \bigg)^{\prod_{j=1}^{n} \bigg[ \big( x_{j} \big) ^{\big[\frac{1}{2^{n-j}} (i - (i \hspace{0.1cm} \text{mod} 2^{n-j})) \big] \text{mod} 2 } \big( 1 - x_{j} \big) ^{\big[\frac{1}{2^{n-j}} ((2^{n}-i-1) - ((2^{n}-i-1) \hspace{0.1cm} \text{mod} 2^{n-j})) \big] \text{mod} 2 } \bigg]}
\end{split}
\end{equation}
\rm{where}
\begin{equation}
\begin{split}
\pi^{n}(i) &= \sum_{j=0}^{2^{m}-1} \prod_{k=0}^{m-1} \pi^{m}(j) \hspace{0.2cm} p^{2^{k+1} \big( j \hspace{0.1cm}  \text{mod} 2^{m-k-1}\big) + \big[ \frac{1}{2^{n-k-1}} (i - (i \hspace{0.1cm} \text{mod} 2^{n-k-1})) \big] \text{mod} 2^{m} } _{2^{k} \big( j \hspace{0.1cm}  \text{mod} 2^{m-k} \big) + \big[ \frac{1}{2^{n-k}} (i - (i \hspace{0.1cm} \text{mod} 2^{n-k})) \big] \text{mod} 2^{m} } \\
& \hspace{1cm} \times \prod_{s=m}^{n-1} p_{\big[ \frac{1}{2^{n-s}} (i - (i \hspace{0.1cm} \text{mod} 2^{n-s})) \big] \text{mod} 2^{m}} ^{\big[ \frac{1}{2^{n-s-1}} (i - (i \hspace{0.1cm} \text{mod} 2^{n-s-1})) \big] \text{mod} 2^{m} }
\end{split}
\end{equation}

\begin{proof}
Let $X$ be a sequence of binary random variables. If $X = X_{1}$, then its probability density function is given by:
\begin{equation}
P(X_{1}=x_{1}) = \pi^{1}(0)^{x_{1}} \times \pi^{1}(1)^{1-x_{1}}.
\end{equation}
If $X = \{X_{1}, X_{2}\}$, then the probability density function is given by:
\begin{equation}
P(X_{1}=x_{1}, X_{2}=x_{2}) = \pi^{2}(00)^{(1-x_{1})(1-x_{2})} \times \pi^{2}(01)^{(1-x_{1})(x_{2})} \times \pi^{2}(10)^{(x_{1})(1-x_{2})} \times \pi^{2}(11)^{(x_{1})(x_{2})}.
\end{equation}
Therefore, for $X = \{X_{1}, X_{2}, \ldots, X_{n}\}$, the probability density function is given by:
\begin{equation}
\begin{split}
\pi^{n}(x_{1}, x_{2}, \ldots, x_{n}) &= \pi^{n}(0...0)^{(1-x_{1})...(1-x_{n})} \times \pi^{n}(0...01)^{(1-x_{1})...(1-x_{n-1})(x_{n})} \times \ldots \\
&\hspace{1cm} \times \pi^{n}(1...10)^{(x_{1})...(x_{n-1})(1-x_{n})} \times \pi^{n}(1...1)^{(x_{1})...(x_{n})} 
\end{split}
\end{equation}
To generate a general expression for $\pi^{n}(x_{1}, x_{2}, \ldots, x_{n})$ we can write the de Bruijn words in terms of their numerical representations. The function is given by the product of the marginal probabilities of all possible length $n$ binary sequences. Hence, we have $\prod_{i=0}^{2^{n}-1} \pi^{n}(i)$. 

This product is raised to the power of a combination of either $x_{j}$ or $(1-x_{j})$ for $i = 1, \ldots, n$, depending on the values of each letter in the sequence. We require $x_{j}$ to occur when $x_{j}=1$ and $1-x_{j}$ to occur when $x_{j}=0$. To achieve this, we have the product of each $x_{j}$ and $(1-x_{j})$ for $j = 1, \ldots, n$, with each term raised to the power $0$ or $1$ according to whether it is present or not.

For each $j$, we require $x_{j}$ to be raised to the power $0$ whenever a $0$ is present in the sequence $i$, and a $1$ whenever a $1$ is present in the sequence $i$. This is equivalent to the logical order of binary sequences of length $n$. Therefore for $j=1$ and $i=0, \ldots, 2^{n-1}-1$, $x_{j}$ should be raised to the power $0$. For $j=1$ and $2^{n-1}, \ldots, 2^{n}-1$, $x_{j}$ should be raised to the power $1$. Then for $j=2$ with $i=0, \ldots, 2^{n-2}-1$ and $2^{n-1}, \ldots, 2^{n-1}+2^{n-2}-1$, $x_{j}$ should be raised to the power $0$. For $j=2$ with $2^{n-2}, \ldots, 2^{n-1}-1$ and $2^{n-1}+2^{n-2}, \ldots, 2^{n}-1$, $x_{j}$ should be raised to the power $1$. This pattern continues until for $j=n$, the power of $x_{j}$ alternates between $0$ and $1$ for each value of $i$. This pattern has general form:
\begin{equation}
\prod_{i=0}^{2^{n}-1} \bigg( \pi^{n}(i) \bigg)^{\prod_{j=1}^{n} \bigg[ \big( x_{j} \big) ^{\big[\frac{1}{2^{n-j}} (i - (i \hspace{0.1cm} \text{mod} 2^{n-j})) \big] \text{mod} 2 } \bigg]}.
\end{equation}

For each $j$, $(1-x_{j})$ must be raised to the opposite power as $x_{j}$ (i.e. if $x_{j}$ is raised to the power $0$, then $(1-x_{j})$ is raised to the power $1$), and so has exactly the same binary pattern. This then gives:
\begin{equation}
\begin{split}
P_{\textbf{X}}(x_{1}&, x_{2}, \ldots, x_{n}) \\
&= \prod_{i=0}^{2^{n}-1} \bigg( \pi^{n}(i) \bigg)^{\prod_{j=1}^{n} \bigg[ \big( x_{j} \big) ^{\big[\frac{1}{2^{n-j}} (i - (i \hspace{0.1cm} \text{mod} 2^{n-j})) \big] \text{mod} 2 } \big( 1 - x_{j} \big) ^{\big[\frac{1}{2^{n-j}} ((2^{n}-i-1) - ((2^{n}-i-1) \hspace{0.1cm} \text{mod} 2^{n-j})) \big] \text{mod} 2 } \bigg]}
\end{split}
\end{equation}

$\pi^{n}(i)$ gives the probability of the length $n$ sequence represented by $i = 0, \ldots, 2^{n}-1$. Using the law of total probability, $\pi^{n}(i)$ can be written in terms of transition probabilities and length $m$ marginal word probabilities:
\begin{equation}
\pi^{n}(i) = \sum_{j=0}^{2^{m}-1} P(i | j) \pi^{m}(j),
\end{equation}
where $j$ represents all possible length $m$ sequences of $0$'s and $1$'s in the numerical representation of the binary words. A general expression for $P(i | j)$, will be a product of $n$ transition probabilities in terms of length $m$ words. The starting word will be the sequence given by $j$, then with each transition, one more letter from the sequence $i$ will included. 

First consider the case when $\pi^{n}(i) = \pi^{n}(0 \ldots 0)$. Using the equation above, the product of transition probabilities will each start with one of the possible $2^{m}$ starting words for the sequence. This word then transitions to a word of the form $*0$, where $*$ is of length $m-1$. The first $m$ transitions will contain letters from the original starting word. The next $n-m$ transitions will then just be of the form $p_{0 \ldots 0}^{0 \ldots 0}$ where a $0$ is added to the sequence at each transition. 

For the first $m-1$ transitions, we start off with a word of the form $*$, then transition to a word of the form $*^{'}0$ where $*^{'}$ is a sequence of length $m-1$. We then transition from $*^{'}0$ to a word of the form $*^{''}00$, where $*^{''}$ is now a length $m-2$ sequence with the first two letters removed from $*$. This pattern continues until the $m^{\text{th}}$ transition which is to a word of the form $0 \ldots 0$. With each transition, there are fewer possible words to transition to and so a cycle is formed. For the $k^{\text{th}}$ transition, there are $2^{m-k}$ possible words to transition to, each of which depends on the word before. Since the transitions cause the word to increase by a factor of two with each $k$, we have the following expression for the first $m-1$ transitions:
\begin{equation}
\pi^{n}(0 \ldots 0) = \sum_{j=0}^{2^{m}-1} \prod_{k=0}^{m-1} \pi^{m}(j) \hspace{0.2cm} p^{2^{k+1} \big( j \hspace{0.1cm}  \text{mod} 2^{m-k-1}\big)} _{2^{k} \big( j \hspace{0.1cm}  \text{mod} 2^{m-k} \big)}
\end{equation}

To extend this to any sequence $i$, each transition now includes an additional letter from $i$. Since the sequences represented by $i$ are all possible length $n$ binary sequences, we have a similar pattern forming to the powers of $x_{j}$ above. For $k=0$, the starting word for the transition is just the initial starting words. Then for $k=1$, the starting word has an additional $0$ for $i=0, \ldots, 2^{n-1}-1$ and an additional $1$ for $i=2^{n-1}, \ldots, 2^{n-1}-1$. This again continues until the last transition alternates between adding a $0$ and a $1$ for $i=0, \ldots, 2^{n}-1$. Since the last $n-m$ transitions simply consist of adding in the letters from the sequence $i$ and the next starting word is the same as the previous word transitioned to, we have the following general expression for $\pi^{n}(i)$:
\begin{equation}
\begin{split}
\pi^{n}(i) &= \sum_{j=0}^{2^{m}-1} \prod_{k=0}^{m-1} \pi^{m}(j) \hspace{0.2cm} p^{2^{k+1} \big( j \hspace{0.1cm}  \text{mod} 2^{m-k-1}\big) + \big[ \frac{1}{2^{n-k-1}} (i - (i \hspace{0.1cm} \text{mod} 2^{n-k-1})) \big] \text{mod} 2^{m} } _{2^{k} \big( j \hspace{0.1cm}  \text{mod} 2^{m-k} \big) + \big[ \frac{1}{2^{n-k}} (i - (i \hspace{0.1cm} \text{mod} 2^{n-k})) \big] \text{mod} 2^{m} } \\
& \hspace{1cm} \times \prod_{s=m}^{n-1} p_{\big[ \frac{1}{2^{n-s}} (i - (i \hspace{0.1cm} \text{mod} 2^{n-s})) \big] \text{mod} 2^{m}} ^{\big[ \frac{1}{2^{n-s-1}} (i - (i \hspace{0.1cm} \text{mod} 2^{n-s-1})) \big] \text{mod} 2^{m} }
\end{split}
\end{equation}
\end{proof}

\section*{Lemma 3.1 (Transition Likelihood, $m=2$)}
For a given sequence, $x = \{x_{1}, x_{2}, \ldots, x_{n}\}$, the joint likelihood in terms of the length $m=2$ de Bruijn transition probabilities $p_{i}^{j}$ is given as follows:
\begin{equation}
\begin{split}
\mathcal{L}(X|p) &= (p_{00}^{00})^{\sum_{i=1}^{n-2} (1-x_{i})(1-x_{i+1})(1-x_{i+2})} \times (p_{00}^{01})^{\sum_{i=1}^{n-2} (1-x_{i})(1-x_{i+1})(x_{i+2})} \\
& \hspace{1cm} \times (p_{01}^{10})^{\sum_{i=1}^{n-2} (1-x_{i})(x_{i+1})(1-x_{i+2})} \times (p_{01}^{11})^{\sum_{i=1}^{n-2} (1-x_{i})(x_{i+1})(x_{i+2})} \\
& \hspace{1cm} \times (p_{10}^{00})^{\sum_{i=1}^{n-2} (x_{i})(1-x_{i+1})(1-x_{i+2})} \times (p_{10}^{11})^{\sum_{i=1}^{n-2} (x_{i})(1-x_{i+1})(x_{i+2})} \\
& \hspace{1cm} \times (p_{11}^{10})^{\sum_{i=1}^{n-2} (x_{i})(x_{i+1})(1-x_{i+2})} \times (p_{11}^{11})^{\sum_{i=1}^{n-2} (x_{i})(x_{i+1})(x_{i+2})} \\
&= (p_{00}^{00})^{n_{00}^{00}} \hspace{0.1cm} (p_{00}^{01})^{n_{00}^{01}} \hspace{0.1cm} (p_{01}^{10})^{n_{01}^{10}} \hspace{0.1cm} (p_{01}^{11})^{n_{01}^{11}} \hspace{0.1cm} (p_{10}^{00})^{n_{10}^{00}} \hspace{0.1cm} (p_{10}^{01})^{n_{10}^{01}} \hspace{0.1cm} (p_{11}^{10})^{n_{11}^{10}} \hspace{0.1cm} (p_{11}^{11})^{n_{11}^{11}} \\
&= (1-p_{00}^{01})^{n_{00}^{00}} \hspace{0.1cm} (p_{00}^{01})^{n_{00}^{01}} \hspace{0.1cm} (1-p_{01}^{11})^{n_{01}^{10}} \hspace{0.1cm} (p_{01}^{11})^{n_{01}^{11}} \hspace{0.1cm} (1-p_{10}^{01})^{n_{10}^{00}} \hspace{0.1cm} (p_{10}^{01})^{n_{10}^{01}} \hspace{0.1cm} \\
& \hspace{1cm} \times (1-p_{11}^{11})^{n_{11}^{10}} \hspace{0.1cm} (p_{11}^{11})^{n_{11}^{11}},
\end{split}
\end{equation}

\begin{proof}
Assume a sequence of letters, $x = \{ x_{1}, x_{2}, ..., x_{n} \}$, where $x_{i} \in [0,1]$ and the ordering is fixed. This sequence can be expressed in terms of its de Bruijn words such that $x = \{w_{1}, w_{2}, ..., w_{n-1}\}$, where $w_{i}$ are the de Bruijn words of length $m=2$. Consider the joint distribution of this sequence. Starting from $w_{1}$, the probability of transitioning to the next word is $p_{w_{1}}^{w_{2}}$. The probability of transitioning to the next following word is, $p_{w_{2}}^{w_{3}}$. This is continued until the transition $p_{w_{n-2}}^{w_{n-1}}$ is reached and produces the joint distribution, $\mathcal{L}(X|p) = p_{w_{1}}^{w_{2}} \times p_{w_{2}}^{w_{3}} \times ... \times p_{w_{n-2}}^{w_{n-1}}$. By collecting like terms for each possible transition probability the above result is given. 
\end{proof}

\section*{Theorem 3.2 (Transition Likelihood, $m \ge 1$)}
For a given sequence, $x = \{x_{1}, x_{2}, \ldots, x_{n}\}$, the joint likelihood in terms of the length $m$ de Bruijn transition probabilities $p_{i}^{j}$ is given as follows:
\begin{equation}
\begin{split}
\mathcal{L}(X|p) &= \left( p_{0 \ldots 0}^{0 \ldots 0} \right)^{\sum_{i=1}^{n-m} (1-x_{i}) \ldots (1-x_{i+m})} \times \left( p_{0 \ldots 00}^{0 \ldots 01} \right)^{\sum_{i=1}^{n-m} (1-x_{i}) \ldots (1-x_{i+m-1})(x_{i+m})} \\
& \hspace{1cm} \times \ldots \times \left( p_{1 \ldots 11}^{1 \ldots 10} \right)^{\sum_{i=1}^{n-m} (x_{i}) \ldots (x_{i+m-1})(1-x_{i+m})} \times \left( p_{1 \ldots 1}^{1 \ldots 1} \right)^{\sum_{i=1}^{n-m} (x_{i}) \ldots (x_{i+m})} \\
&= \left( p_{0 \ldots 0}^{0 \ldots 0} \right)^{n_{0 \ldots 0}^{0 \ldots 0}} \times \left( p_{0 \ldots 00}^{0 \ldots 01} \right)^{n_{0 \ldots 00}^{0 \ldots 01}} \times \ldots \times \left( p_{1 \ldots 11}^{1 \ldots 10} \right)^{n_{1 \ldots 11}^{1 \ldots 10}} \times \left( p_{1 \ldots 1}^{1 \ldots 1} \right)^{n_{1 \ldots 1}^{1 \ldots 1}} \\
&= \prod_{i=0}^{2^{m+1}-1} \left(p_{\frac{1}{2} (i - (i \hspace{0.1cm} \text{mod } 2))}^{i \hspace{0.1cm} \text{mod } 2^{m}} \right) ^{n_{\frac{1}{2} (i - (i \hspace{0.1cm} \text{mod } 2))}^{i \hspace{0.1cm} \text{mod } 2^{m}}}\\
&= \prod_{i=0}^{2^{m}-1} \left( 1 - p_{i}^{(2i+1) \hspace{0.1cm} \text{mod } 2^{m}} \right)^{n_{i}^{((2i+1) \hspace{0.1cm} \text{mod } 2^{m}) - 1}} \left( p_{i}^{(2i+1) \hspace{0.1cm} \text{mod } 2^{m}} \right) ^{n_{i}^{((2i+1) \hspace{0.1cm} \text{mod } 2^{m})}}.
\end{split}
\end{equation}

\begin{proof}
Assume a sequence of letters, $x = \{x_{1}, x_{2}, ..., x_{n}\}$, where $x_{i} \in [0,1]$ and the ordering is fixed. This sequence can be written in terms of its de Bruijn words such that $x = w_{1}, w_{2}, ..., w_{n-1}$, where $w_{i}$ are the de Bruijn words of length $m$. Consider the joint distribution of this sequence. Starting from $w_{1}$, the probability of transitioning to the next word is $p_{w_{1}}^{w_{2}}$. The probability of transitioning to the next following word is, $p_{w_{2}}^{w_{3}}$. This continues until the transition $p_{w_{n-2}}^{w_{n-1}}$ which gives the joint distribution, $\mathcal{L}(X|p) = p_{w_{1}}^{w_{2}} \times p_{w_{2}}^{w_{3}} \times ... \times p_{w_{n-2}}^{w_{n-1}}$. Like terms are collected for each possible transition probability. 

There are $2^{m+1}$ possible transition probabilities since each word can be followed by either a $0$ or a $1$. Since rows in the transition matrix sum to one, the transition likelihood can be expressed in terms of $2^{m}$ parameters of the form $*1$. These are expressed as $2i + 1$ for $i \in (0:2^{m} -1)$ (using the numerical representation of the binary words). Since all possible words transition to a word of the form $*1$, where $*$ is of length $m-1$, the possible end transition words must occur twice and a repeated cycle is formed. This occurs every $2^{m}$ transitions, hence:
\begin{equation}
\mathcal{L}(X|p) = \prod_{i=0}^{2^{m}-1} \left( 1 - p_{i}^{((2i+1) \hspace{0.1cm} \text{mod } 2^{m})} \right)^{n_{i}^{((2i+1) \hspace{0.1cm} \text{mod } 2^{m}) - 1}} \left( p_{i}^{((2i+1) \hspace{0.1cm} \text{mod } 2^{m})} \right) ^{n_{i}^{((2i+1) \hspace{0.1cm} \text{mod } 2^{m})}}.
\end{equation}
\end{proof}

\section*{Theorem 3.3 (Fisher Information, $m \ge 1$)}
The Fisher information, $I(p_{k}) = -E\left[\frac{\partial^{2}\text{log}\mathcal{L}}{\partial {p_{k}}^{2}}\right]$, for each transition probability $p_{k}$ for and sequence $n$ and word length $m$ is given by:
\begin{equation}
\begin{split}
I(p_{k}) &= -E\left[\frac{\partial^{2}\text{log}\mathcal{L}}{\partial {p_{k}}^{2}}\right] \\
&= \frac{1}{p_{k}} \sum_{i=0}^{n-m-1} \sum_{j=0}^{2^{n-m-1}-1} \pi^{n} \left( 2^{m+1} j + 2^{i}k - \left( 2^{m+1} - 1 \right) \left( j \hspace{0.1cm} \text{mod} 2^{i} \right) \right)
\end{split}
\end{equation}

\begin{proof}
For transition probabilities, $p$, the Fisher information is given by:
\begin{equation}
I(p) = -E \left[ \frac{\partial^{2} \text{log}\mathcal{L}(p|x)}{\partial p^{2}} \right],
\end{equation}
where $\mathcal{L}(p|x)$ is given as
\begin{equation}
\begin{split}
\mathcal{L}(X|p) &= \left( p_{0 \ldots 0}^{0 \ldots 0} \right)^{\sum_{i=1}^{n-m} (1-x_{i}) \ldots (1-x_{i+m})} \times \left( p_{0 \ldots 00}^{0 \ldots 01} \right)^{\sum_{i=1}^{n-m} (1-x_{i}) \ldots (1-x_{i+m-1})(x_{i+m})} \\
& \hspace{1cm} \times \ldots \times \left( p_{1 \ldots 11}^{1 \ldots 10} \right)^{\sum_{i=1}^{n-m} (x_{i}) \ldots (x_{i+m-1})(1-x_{i+m})} \times \left( p_{1 \ldots 1}^{1 \ldots 1} \right)^{\sum_{i=1}^{n-m} (x_{i}) \ldots (x_{i+m})} \\
&= \prod_{i=0}^{2^{m+1}-1} \left(p_{\frac{1}{2} (i - (i \hspace{0.1cm} \text{mod } 2))}^{i \hspace{0.1cm} \text{mod } 2^{m}} \right) ^{n_{\frac{1}{2} (i - (i \hspace{0.1cm} \text{mod } 2))}^{i \hspace{0.1cm} \text{mod } 2^{m}}},
\end{split}
\end{equation}
from Theorem 4.2.

From $\mathcal{L}(p|x)$, $\text{log}\mathcal{L}$ is given as 
\begin{equation}
\begin{split}
\text{log}\mathcal{L}(p|x) &= \sum_{i=1}^{n-m} (1-x_{i}) \ldots (1-x_{i+m}) \text{log} \left(p_{0 \ldots 0}^{0 \ldots 0} \right) \\
& \hspace{0.5cm} + \sum_{i=1}^{n-m} (1-x_{i}) \ldots (1-x_{i+m-1})(x_{i+m}) \text{log} \left(p_{0 \ldots 0}^{0 \ldots 0} \right) \\
& \hspace{1cm} + \ldots + \sum_{i=1}^{n-m} (x_{i}) \ldots (x_{i+m}) \text{log} \left(p_{1 \ldots 1}^{1 \ldots 1} \right) \\
&= n_{0 \ldots 0}^{0 \ldots 0} \text{log} \left(p_{0 \ldots 0}^{0 \ldots 0} \right) + n_{0 \ldots 00}^{0 \ldots 01} \text{log} \left(p_{0 \ldots 0}^{0 \ldots 0} \right) + \ldots + n_{1 \ldots 1}^{1 \ldots 1} \text{log} \left(p_{1 \ldots 1}^{1 \ldots 1} \right) \\
&= \sum_{i=0}^{2^{m+1}-1} n_{\frac{1}{2} (i - (i \hspace{0.1cm} \text{mod } 2))}^{i \hspace{0.1cm} \text{mod } 2^{m}} \text{log} \left(p_{\frac{1}{2} (i - (i \hspace{0.1cm} \text{mod } 2))}^{i \hspace{0.1cm} \text{mod } 2^{m}} \right)
\end{split}
\end{equation}

Let $p_k$ be the $k^{\text{th}}$ transition probability from logical order, $\left\{ p_{0 \ldots 0}^{0, \ldots, 0}, p_{0 \ldots 00}^{0 \ldots 01}, \ldots, p_{1 \ldots 11}^{1 \ldots 10}, p_{1 \ldots 1}^{1 \ldots 1} \right\}$. Then $\frac{\partial^{2}\text{log}\mathcal{L}}{\partial {p_{k}}^{2}}$ for transition probability $p_{k}$ is given as the following:
\begin{equation}
\frac{\partial^{2}\text{log}\mathcal{L}}{\partial {p_{k}}^{2}} = -\frac{1}{p_{k}^{2}} \, n_{p_{k}},
\end{equation}
where $n$ can be written in terms of $x = \{x_{1}, \ldots, x_{n}\}$ as above. For example,
\begin{equation}
\frac{\partial^{2}\text{log}\mathcal{L}}{\partial {p_{0 \ldots 0}^{0 \ldots 0}}^{2}} = -\frac{1}{{p_{0 \ldots 0}^{0 \ldots 0}}^{2}} \sum_{i=1}^{n-m} (1-x_{i}) \ldots (1-x_{i+m}).
\end{equation}

An expectation with respect to $x$ is given as:
\begin{equation}
E[g(X_{1}, \ldots, X_{n})] = \sum_{x_{1}=0}^{x_{1}=1} \cdots \sum_{x_{n}=0}^{x_{n}=1} g(x_{1}, \ldots, x_{n}) P_{\textbf{X}}(x_{1}, \ldots, x_{n}),
\end{equation}
where $P_{\textbf{X}}(x_{1}, \ldots, x_{n})$ is as defined from Theorem 3.1.
For the above example, the expectation of the double integral is as follows:
\begin{equation}
\begin{split}
E\left[\frac{\partial^{2}\text{log}\mathcal{L}}{\partial {p_{0 \ldots 0}^{0 \ldots 0}}^{2}}\right] &= -\frac{1}{{p_{0 \ldots 0}^{0 \ldots 0}}^{2}} \Big( E[(1-x_{1}) \ldots (1-x_{m+1})] + E[(1-x_{2}) \ldots (1-x_{m+2})] \\
& \hspace{1cm} + \ldots + E[(1-x_{n-m}) \ldots (1-x_{n})] \Big) \\
&= -\frac{1}{{p_{0 \ldots 0}^{0 \ldots 0}}^{2}} \Bigg( \sum_{x_{1}=0}^{x_{1}=1} \cdots \sum_{x_{m+1}=0}^{x_{m+1}=1} (1-x_{1}) \ldots (1-x_{m+1}) P_{\textbf{X}}(x_{1}, \ldots, x_{n}) \\
& \hspace{1cm} + \ldots + \sum_{x_{n-m}=0}^{x_{n-m}=1} \cdots \sum_{x_{n}=0}^{x_{n}=1} (1-x_{n-m}) \ldots (1-x_{n}) P_{\textbf{X}}(x_{1}, \ldots, x_{n}) \Bigg) \\
&= -\frac{1}{{p_{0 \ldots 0}^{0 \ldots 0}}^{2}} \bigg( \pi^{n}(0 \ldots 0) + \pi^{n}(0 \ldots 01) + \ldots + \pi^{n}(1 \ldots 10 \ldots 0) \bigg).
\end{split}
\end{equation}
The last line from the expression above comes from substituting in each $x_{i}$ for $i = 1, \ldots, n$. This ensures that the result is the sum of all of the possible marginal probabilities of length $n$ sequences that contain the given length $m+1$ transition sequence.

We can then form a general expression for this sum. If $p_{k} = p_{0 \ldots 0}^{0 \ldots 0}$, then a sequence of length $m+1$ $0$'s must be contained in each marginal probability sequence.

For each sequence length $n$ there will be $n-m$ different summations for each combination of $m+1$ letters. There are $n-m$ places where the length $m+1$ transition sequence can occur in the full length $n$ sequence. Each of these then has $2^{n-m}$ terms as there are then $n-m$ letters remaining that can take any value and we must include all combinations. Hence overall, there will be $(n-m) \times (2^{n-m})$ terms in the general expression.

If the length $m+1$ sequence of $0$'s are in the far most right position in the sequence, then the remaining letters must start in the $m+1$ position from the right corresponding to $2^{m+1}$ in binary notation. All of the possible words for this combination consist of multiples of $2^{m+1}$ since as we increase the number of $1$'s in binary order to the sequence, we add an extra $2^{m+1}$ each time. This gives $\sum_{j=0}^{2^{n-m-1}} \pi^{n}(2^{m+1} j)$. 

If the $m+1$ sequence of $0$'s then moves one place to the left, the remaining letters now lie either side of the transition word. The $m+1$ position from the right now contains the first letter from the transition sequence and so $2^{m+1}$ must be removed from the decimal representation of the binary sequence. If the sequence moves further to the left, then more positions will contain the transition sequence that we then must remove accordingly. We only want to remove the corresponding amounts when a $1$ would occur in the sequence in place of the transition sequence. This means that a cycle is formed. Considering all the binary sequences of length $n$ where the binary position is indexed by $i=1,2,\ldots$ (from right to left), $1$'s occur in groups of length $2^{i}$. This corresponds to a cycle as follows: $\sum_{i=0}^{n-m-1} \sum_{j=0}^{2^{n-m-1}-1} \pi^{n} (2^{m+1} j - 2^{m+1} j \text{ mod } 2^{i})$.

Lastly, we must include the remaining letters that appear in the right hand positions as the transition sequence moves with increasing $i$. This occurs in the same cycle as above given by: $\sum_{i=0}^{n-m-1} \sum_{j=0}^{2^{n-m-1}-1} \pi^{n} (j \text{ mod } 2^{i})$. Putting each of these sections together gives the following general result:
\begin{equation}
\begin{split}
E\left[\frac{\partial^{2}\text{log}\mathcal{L}}{\partial {p_{0 \ldots 0}^{0 \ldots 0}}^{2}}\right] &= -\frac{1}{{p_{0 \ldots 0}^{0 \ldots 0}}^{2}} \sum_{i=0}^{n-m-1} \sum_{j=0}^{2^{n-m-1}-1} \pi^{n} (2^{m+1} j  + j \text{ mod } 2^{i} - 2^{m+1} j \text{ mod } 2^{i}) \\
&= -\frac{1}{{p_{0 \ldots 0}^{0 \ldots 0}}^{2}} \sum_{i=0}^{n-m-1} \sum_{j=0}^{2^{n-m-1}-1} \pi^{n} (2^{m+1} j - (2^{m+1}-1) \, j \text{ mod } 2^{i}).
\end{split}
\end{equation}

This is the general form when the transition is from a word of the form $0 \ldots 0$ to a word of the same form. To be applicable for any transition we need to include the additional letters. This just equates to adding in the correct numerical representation for the length $m+1$ transition sequence. If these are listed in order as above for $k$, this is given by $\sum_{i=0}^{n-m-1} \sum_{j=0}^{2^{n-m-1}-1} \pi^{n} ( 2^{i}k )$.

Therefore, putting all of these steps together the fisher information is given as the following:
\begin{equation}
\begin{split}
I(p_{k}) &= -E\left[\frac{\partial^{2}\text{log}\mathcal{L}}{\partial {p_{k}}^{2}}\right] \\
&= \frac{1}{p_{k}} \sum_{i=0}^{n-m-1} \sum_{j=0}^{2^{n-m-1}-1} \pi^{n} \left( 2^{m+1} j + 2^{i}k - \left( 2^{m+1} - 1 \right) \left( j \hspace{0.1cm} \text{mod} 2^{i} \right) \right)
\end{split}
\end{equation}
\end{proof}

\section*{Theorem 3.4 (Posterior Distribution for de Bruijn Probability Transitions, $m \ge 1$)}
Given Bayes' theorem, the posterior distribution of the de Bruijn transition probabilities is expressed below:
\begin{equation}
\begin{split}
P(p | X) & = \frac{\mathcal{L}(X|p,m) P(p|m)}{P(X)} \\
& = \frac{\mathcal{L}(X|p,m)P(p|m)}{\int \mathcal{L}(X|p,m)P(p|m) dp}
\end{split}
\end{equation}
\rm{where,}
\begin{equation}
\begin{split}
\mathcal{L}(X|p,m)P(p|m) &= \prod_{i=0}^{2^{m}-1} (1 - p_{i}^{((2i+1) \hspace{0.1cm} \text{mod } 2^{m})})^{n_{i}^{((2i+1) \hspace{0.1cm} \text{mod } 2^{m})-1 }+\beta_{i+1}-1} \\
& \hspace{2cm} \times (p_{i}^{((2i+1) \hspace{0.1cm} \text{mod } 2^{m})})^{n_{i}^{((2i+1) \hspace{0.1cm} \text{mod } 2^{m}) }+\alpha_{i+1}-1}
\end{split}
\end{equation}
\rm{and}
\begin{equation}
\int P(X|p,m)P(p|m) dp = \prod_{i=0}^{2^{m}-1} \frac{\Gamma(n_{i}^{((2i+1) \hspace{0.1cm} \text{mod } 2^{m}) - 1} + \beta_{i+1})\Gamma(n_{i}^{((2i+1) \hspace{0.1cm} \text{mod } 2^{m})}) + \alpha_{i+1})}{\Gamma(n_{i}^{((2i+1) \hspace{0.1cm} \text{mod } 2^{m}) - 1} + n_{i}^{((2i+1) \hspace{0.1cm} \text{mod } 2^{m})} +\beta_{i+1}+\alpha_{i+1})}
\end{equation}

\begin{proof}
The equation:
\begin{equation}
\begin{split}
P(p | X) & = \frac{\mathcal{L}(X|p,m) P(p|m)}{P(X)} \\
& = \frac{\mathcal{L}(X|p,m)P(p|m)}{\int \mathcal{L}(X|p,m)P(p|m) dp}
\end{split}
\end{equation}
is taken from Bayes' theorem where the likelihood is given in Theorem 4.2 and the prior is defined to be a product of beta densities: 
\begin{equation}
P(p|m) = \prod_{i=0}^{2^{m}-1} \frac{\Gamma(\alpha_{i+1} + \beta_{i+1})}{\Gamma(\alpha_{i+1})\Gamma(\beta_{i+1})}p^{\alpha_{i+1} -1}(1-p)^{\beta_{i+1} -1},
\end{equation}
for unknown parameters $\alpha$ and $\beta$. Substituting this in gives:
\begin{equation}
\begin{split}
P(p | X) & \propto \mathcal{L}(X|p,m)P(p|m) \\
& = \prod_{i=0}^{2^{m}-1} \Bigg[ \left( 1 - p_{i}^{((2i+1) \hspace{0.1cm} \text{mod } 2^{m})} \right)^{n_{i}^{((2i+1) \hspace{0.1cm} \text{mod } 2^{m}) - 1}} \left( p_{i}^{((2i+1) \hspace{0.1cm} \text{mod } 2^{m})} \right) ^{n_{i}^{((2i+1) \hspace{0.1cm} \text{mod } 2^{m})}} \\
& \hspace{2cm} \times \left( 1 - p_{i}^{((2i+1) \hspace{0.1cm} \text{mod } 2^{m})} \right)^{\beta_{i+1}-1} \left( p_{i}^{((2i+1) \hspace{0.1cm} \text{mod } 2^{m})} \right) ^{\alpha_{i+1}-1} \Bigg] \\
& = \prod_{i=0}^{2^{m}-1} (1 - p_{i}^{((2i+1) \hspace{0.1cm} \text{mod } 2^{m})})^{n_{i}^{((2i+1) \hspace{0.1cm} \text{mod } 2^{m}) - 1}+\beta_{i+1}-1} \\
& \hspace{2cm} \times (p_{i}^{((2i+1) \hspace{0.1cm} \text{mod } 2^{m})})^{n_{i}^{((2i+1) \hspace{0.1cm} \text{mod } 2^{m}) }+\alpha_{i+1}-1}.
\end{split}
\end{equation}

Both the prior, $P(p|m)$, and the posterior, $P(p|X)$, take the form of a product of beta densities, hence there is a conjugate relationship and the following is true:
\begin{equation}
\int P(S|p,m)P(p|m) dp = \prod_{i=0}^{2^{m}-1} \frac{\Gamma(n_{i}^{((2i+1) \hspace{0.1cm} \text{mod } 2^{m}) - 1} + \beta_{i+1})\Gamma(n_{i}^{((2i+1) \hspace{0.1cm} \text{mod } 2^{m})}) + \alpha_{i+1})}{\Gamma(n_{i}^{((2i+1) \hspace{0.1cm} \text{mod } 2^{m}) - 1} + n_{i}^{((2i+1) \hspace{0.1cm} \text{mod } 2^{m})} +\beta_{i+1}+\alpha_{i+1})}.
\end{equation}
\end{proof}

\end{document}